\documentclass[conference]{IEEEtran}

\ifCLASSINFOpdf
\else

\fi

\DeclareMathAlphabet{\mathpzc}{OT1}{pzc}{m}{it}  

\usepackage{wrapfig}
\usepackage[utf8]{inputenc}
\usepackage{tikz}
\usetikzlibrary{positioning}
 \usepackage[utf8]{inputenc}
\usepackage[pdfencoding=auto]{hyperref}
\usepackage{amsmath}
\usepackage{mathrsfs}  
\usepackage{times} 
\usepackage{sidecap}
\usepackage{graphicx} 
\graphicspath{{figures/}} 
\usepackage{booktabs} 
\usepackage[font=small,labelfont=bf]{caption} 
\usepackage{amsfonts, amsmath, amsthm, amssymb} 
\usepackage{wrapfig}
\usepackage[utf8]{inputenc}
\usepackage{tikz}
\usepackage{array}
\usepackage{tabu}

\font\msbm=msbm10 at 10pt

\newcommand{\FF}{\mbox{\msbm F}}

\def \F {{\FF}}

\newtheorem{theorem}{Theorem}

\newtheorem{remark}{Remark}
\newtheorem{example}{Example}
\newtheorem{definition}{Definition}

\begin{document}
%
\title{On Some Universally Good Fractional Repetition Codes}

\author{
\IEEEauthorblockN{Shreyansh A. Prajapati and Manish K. Gupta\\}
\IEEEauthorblockA{Laboratory of Natural Information Processing \\ Dhirubhai Ambani Institute of Information and Communication Technology
Gandhinagar, Gujarat, 382007 India\\
Email: shreyansh\_prajapati@daiict.ac.in, mankg@computer.org}
}

\maketitle

\begin{abstract}
Data storage in \textit{Distributed Storage Systems} (DSSs) is a multidimensional optimization problem. Using network coding, one wants to provide 
reliability, scalability, security, reduced storage overhead, reduced bandwidth for repair and minimal disk I/O etc. in such systems. Regenerating codes have been used to optimize some of these parameters, where a file can be reconstructed by contacting any $k$ nodes in the system and in case of node failure it can be repaired by using any $d$ nodes in the system. This was further generalized to Fractional repetition (FR) codes (a smart replication of encoded packets) on $n$ nodes which also provides optimized disk I/O and where a node failure can be repaired by contacting some specific set of nodes in the system. Several constructions of FR codes using graphs and combinatorial designs are known. In particular,  some constructions of codes for heterogeneous DSSs are given using partial regular graph (where number of packets on each node is different) and ring construction. In this work, we show that the codes constructed using the partial regular graph are universally good code. Further, we found several universally good codes using ring construction and $t$-construction.  

\end{abstract}

\IEEEpeerreviewmaketitle


\section{Introduction}
The \textit{cloud} word is very trending in the 21st century. Cloud is more frequently used as a data storage device. Cloud, also known as a Distributed Storage Systems(DSSs), is a connected network of unreliable data storage devices where data of an individual or organization can be stored. After the break through discovery of Internet watching videos on Youtube, photos of your friends on Facebook, collaborating to work with your colleagues on Google DOC, reading a newspapers online, studying various courses using tutorials and presentations is very easy. Distributed Distributed Storage Systems (DSS) are widely used to store a large amount of data with high reliability for a long period of time using a distributed collection of storage nodes in large data centers. Individual data nodes can be unreliable \cite{Ghemawat03}. 

Traditional DSSs are using repetition codes to store data \cite{5496972}. Source file will be divided into chunks of certain size and stored on the different data nodes on the DSS. Each chunk known as \textit{packets} has some encoded fraction of information of the source file. In order to provide reliability of the data, most of the DSSs replicate the packets several times. Replication scheme is the easiest scheme to ensure reliability but storing the data in this manner will increase the storage overhead. In case of node failure, new node will copy data from its replication.  The total amount of data packets to be downloaded to repair a failed node is known as \textit{Repair bandwidth}. Total number of nodes from which the packets are downloaded is called the \textit{Repair degree}. Repair can be \textit{exact} where exact copy of the lost packet is copied on the new node or \textit{functional} where new packet is the function of the lost packet. Erasure coding schemes were introduced for enhanced storage efficiency. Classical erasure codes give greater degree of data reliability with small storage overhead as compared to Replication scheme. For that reason, giant market players in cloud storage like Facebook Analytics Hadoop Cluster \cite{XorbasVLDB}, Microsoft Windows Azure \cite{huang2012erasure} are using erasure coding techniques for improved performance and low storage overhead. In erasure coding scheme, the file of size $B$ is divided into k encoded message packets of size $B$/$k$. The symbols of each packets are from field $\F_q$. These packets are encoded using an ($n$, $k$) Maximum Distance Separable (MDS) codes and encoded packets are spread on $n$ distinct nodes in such a manner that any $k$ (called reconstruction degree) nodes have sufficient encoded packets to reconstruct the complete source file. As mentioned in  \cite{XorbasVLDB} Reed-Solomon codes have been implemented for encoding of $8\%$ of total stored data on the top of Facebook HDFS. $20\%$ of the network traffic was generated because of the repair traffic of that $8\%$ of encoded data. Erasure codes are not optimal for node repairs. To optimize the two conflicting parameters viz. data storage space and repair bandwidth in a seminal paper, Dimakis et al. have introduced \textit{regenerating codes} \cite{5550492}. The trade-off curve between the node storage capacity and the repair bandwidth is analysed for the regenerating codes. Trade-off curve between node storage capacity and repair bandwidth have been studied by many researchers for various DSSs \cite{6620424,DBLP:journals/corr/BenerjeeG15,ETT:ETT2887,Akhlaghi20102105,5550492}. On the trade-off curve, codes associated with the optimal parameters are known as minimum storage regenerating (MSR) and minimum bandwidth regenerating (MBR) codes. Construction of MBR codes with exact uncoded repair is given in \cite{RSKR09}. These codes reduces the disk I/O operations required to repair a failed node. Extension of their work led to the idea of the Fractional Repetition (FR) code \cite{rr10}. 

The FR code consist of two layer coding, inner repetition code and outer MDS code \cite{rr10}. On ($n,k,d$) DSS with identical repair traffic $\beta$, FR code can be characterized by the parameters $n,\theta,\alpha$ and $\rho$, where replicas of $\theta$ packets (each replicated $\rho$ times) are stored among $n$ distinct nodes having node storage capacity $\alpha$ each. Construction of the particular FR code in the particular work is done by Regular graph and Steiner system \cite{rr10}. The conditions on parameters for the existence of such FR codes are investigated in \cite{DBLP:journals/corr/abs-1201-3547}. Constructions of FR codes defined on ($n,k,d$) DSS by bipartite graphs \cite{6120326}, resolvable designs \cite{6483351}, Kronecker product of two Steiner systems \cite{6810361}, randomize algorithm \cite{6033980}, incidence matrix \cite{DBLP:journals/corr/abs-1303-6801} are also investigated. FR codes on heterogeneous DSS are also investigated by researchers \cite{6804948, 6763122, 7118709, iet:/content/journals/10.1049/iet-com.2014.1225}. In \cite{6804948}, \textit{Irregular Fractional Repetition code} (IFR code) is constructed by uniform hypergraph. In particular, IFR code have non-uniform node storage capacity $\alpha_i$ ($\forall\; i\in[n]$) and failed node is repaired by some $d$ nodes. In \cite{6763122}, \textit{General Fractional Repetition code} (GFR code) is constructed by group divisible designs. The specific GFR code is FR code with repair traffic $\beta\leq 1$ and dynamic node storage capacity $\alpha_i$ ($\forall\; i\in[n]$). FR code with dynamic repair degree called \textit{Weak Fractional Repetition code} (WFR code) is constructed by partial regular graph in \cite{DBLP:journals/corr/abs-1302-3681}.
FR code on Batch code is investigated in \cite{7118709}. In \cite{iet:/content/journals/10.1049/iet-com.2014.1225}, FR code is constructed such that replication factor of an arbitrary packet $P_j$ $(j\in[\theta])$ is either $\rho_1$ or $\rho_2$. Construction of FR code based on sequences, is investigated in \cite{KGBenerjee}.

In \cite{5707092} researchers have generalized the construction of \cite{Rashmi:2009:ECO:1793974.1794188} and introduced DRESS \textit{(Distributed Replication based Exact Simple Storage)} codes \cite{6033980} which consist of the concatenation of an outer MDS code and an inner repetition code called fractional repetition (FR) code.

The paper is organized as follows. In Section 2, we summarize the bounds for universally good codes for both homogeneous and heterogeneous distributed storage systems as discussed in \cite{6763122}. We also give the theorem for reconstruction degree for codes constructed using partial regular graph as described in \cite{DBLP:journals/corr/abs-1302-3681} and prove that these FR codes are universally good. In Section 3, we provide the generalization of ring construction discussed in \cite{DBLP:journals/corr/abs-1302-3681} for $\rho \geq 2$. Finally in Section 4, we give the parameters of universally good FR codes constructed using simple construction studied in Section 3 of \cite{DBLP:journals/corr/abs-1302-3681} for $n = \theta$ and $\rho = d$.

\begin{definition}[Fractional Repetition Code]
 On a distributed storage system with $n$ nodes denoted by $U_i$ ($i \in [n]$) and $\theta$ packets denoted by $P_j$ ($j \in [\theta]$), FR code  $\mathscr{C}$($n, \theta, \alpha, \rho$) is a collection of $n$ subsets $U_i$ ($i\in[n]$) of a set \{$P_j:j\in[\theta]$\}, such that $\forall\ j\in[\theta]$, packet $P_j$ is an element of exactly $\rho$ distinct subsets in the collection $\mathscr{C}$, where $|U_i|=\alpha_i$ which denotes the number of packets stored on the node $U_i$, $\alpha$=$\max_{i\in[n]}\{\alpha_i\}$.
 \label{FR code}
\end{definition}

\begin{example}
Figure-\ref{2} is an example of FR codes for DSS with parameters $\left(n=7, d=5, k=5,\theta=17,\rho=2\right)$. In the Figure-\ref{2}, 17 packets are stored among $U_1, U_2, ..., U_7$ nodes.
\end{example}

$(n,k,d)$ DSSs with $[\theta,M(k)]$ outer MDS code and $(n,\alpha,\rho)$ inner FR code is \textit{universally good} if $\forall$ $ k \leq \alpha$ the DRESS code satisfy the following property:
\begin{equation}
M(k) \geq k\alpha - {k \choose 2}.
\label{eq1}
\end{equation}
 Other bounds for optimal codes are studied in \cite{7118709}.

In \cite{6259860, 6284027, 5934901} researchers have introduced the design of DSSs for local repair $d \leq k$. Construction of \textit{locally repairable FR codes} is discussed in \cite{DBLP:journals/corr/OlmezR14},\cite{7458387}. In this paper we give the construction for FR codes for $d > k$.


\begin{definition}$(\left(d_1,d_2,...,d_m \right)$-Regular Graph):\\
A $\left(d_1,d_2,...,d_m \right)$-Regular Graph $G(V,E)$ is an undirected graph in which $\mid V \mid = n$ and $ \;\exists\;$ $V_1,V_2,...,V_m$ $\subset$ $V$,  $\sum_{i=1}^{\theta} \mid V_i \mid = n$ and $V_i \cap V_j = \phi$,$\left(1 \leq i,j \leq \theta\right)$. All the vertices in the set $V_i$ has degree $d_i$, $\left(1 \leq i \leq m\right)$. If degree of vertices $V_1, V_2, ..., V_n$ is $d_1, d_2, ..., d_n$ respectively then $$\sum_{i=1}^{n} d_i = 2\mid E \mid$$.

\begin{remark}
$d$-Regular graph is a special case of $\left(d_1,d_2,...,d_m\right)$-Regular graph where $d_i = d$ $\forall i \in[\theta]$.
\end{remark}
\end{definition}

Partial regular graph discussed in \cite{DBLP:journals/corr/abs-1302-3681} is a $\left(d_1,d_2,...,d_\theta \right)$-Regular Graph where some $n-1$ vertices has degree $d$ and one vertex has degree $d-1$.Here $n$ and $d$ is odd positive numbers. In \cite{DBLP:journals/corr/abs-1302-3681} researchers have discussed about the construction partial regular graph and construction of  weak FR codes based on the partial regular graph. In this paper we identify that for a given parameter whether constructed WFR codes are \textit{universally good} or not.

\begin{example}
Figure \ref{1} is an example of the partial regular graph PRG($n = 7, d = 5$), where \textit{degree} of node $U_7$ is d-1 and remaining all the vertexes have degree d. Figure \ref{2} is the node packet distribution of FR codes constructed using PRG($n = 7, d = 5$) in Figure \ref{1}.
\end{example}

\begin{theorem}
For $(n,k,d)$ DSS with $[\theta,\theta-1]$ MDS codes, FR codes constructed using partial regular graph has reconstruction degree $k$ = $n-2$.
\label{th1}
\end{theorem}

\begin{proof}
Observe that if we remove any $n-2$ vertices of PRG(n,d) then the resultant graph is a complete graph with 2 vertices $K_2$. Only one edge will remain present which means that by contacting any $n-2$ nodes out of $n$ nodes one can get $\theta-1$ distinct packets required to reconstruct the entire file.
\end{proof}

If $k$ is the reconstruction degree then the lower bound on distinct number of packets that can be downloaded by contacting any $k$ nodes with variable repair degree is discussed in Theorem 4 of \cite{6763122}. Weak FR code constructed using PRG ($n,d$) has $d$ packets on $n-1$ nodes and $d-1$ packet on one node. So WFR codes based on the partial regular graph are universally good if the following property is satisfied:
\begin{equation}
M(k) \geq k\alpha - {k \choose 2}-1
\label{par_equ}
\end{equation}
 
\begin{theorem}
WFR codes constructed using partial regular graph are universally good.
\end{theorem}

\begin{proof}
Since we are using $[\theta,\theta-1]$ outer MDS code $M(k)=\theta-1$. As per the Equation (\ref{par_equ}) we have to prove that WFR constructed using partial regular graph satisfy the following property:\\
\begin{equation}
\theta-1 \geq k\alpha - {k \choose 2}-1 
\label{part}
\end{equation}\\
Since number of nodes $n$ and $d$ is odd in partial regular graph, take $n = 2p+1$ and $d = 2q+1$, where $p,q \in N,p>q.$ As discussed in Theorem \ref{th1}, $k = n-2 = 2p-1$, $\rho = 2$. Property of FR codes was discussed in \cite{5707092} is as follows:
\begin{equation}
nd = \rho \theta
\label{eq_4}
\end{equation} 
Equation-(\ref{eq_4}) is true for any FR codes with replication factor $\rho$ for $(n,k,d)$ DSSs. Since degree of one node is $d-1$ WFR constructed using partial regular graph will satisfy $nd-1 = \rho \theta$. So $\theta = 2pq+p+q$. If we substitute all the values in Equation (\ref{part}) and simplify it then we'll get following non-linear equation of two variables:
  
\[
2p^2-2pq-4p+3q+2 \geq 0.
\]

Which is always true for any $p$ and $q$,$p,q \in N$,where $p>q$. WFR constructed using partial regular graph satisfy the property mentioned in Equation (\ref{par_equ}). Hence WFR codes constructed using partial regular graph is \textit{universally good}.
\end{proof}

\begin{figure}
    \centering
    \includegraphics[scale=0.3]{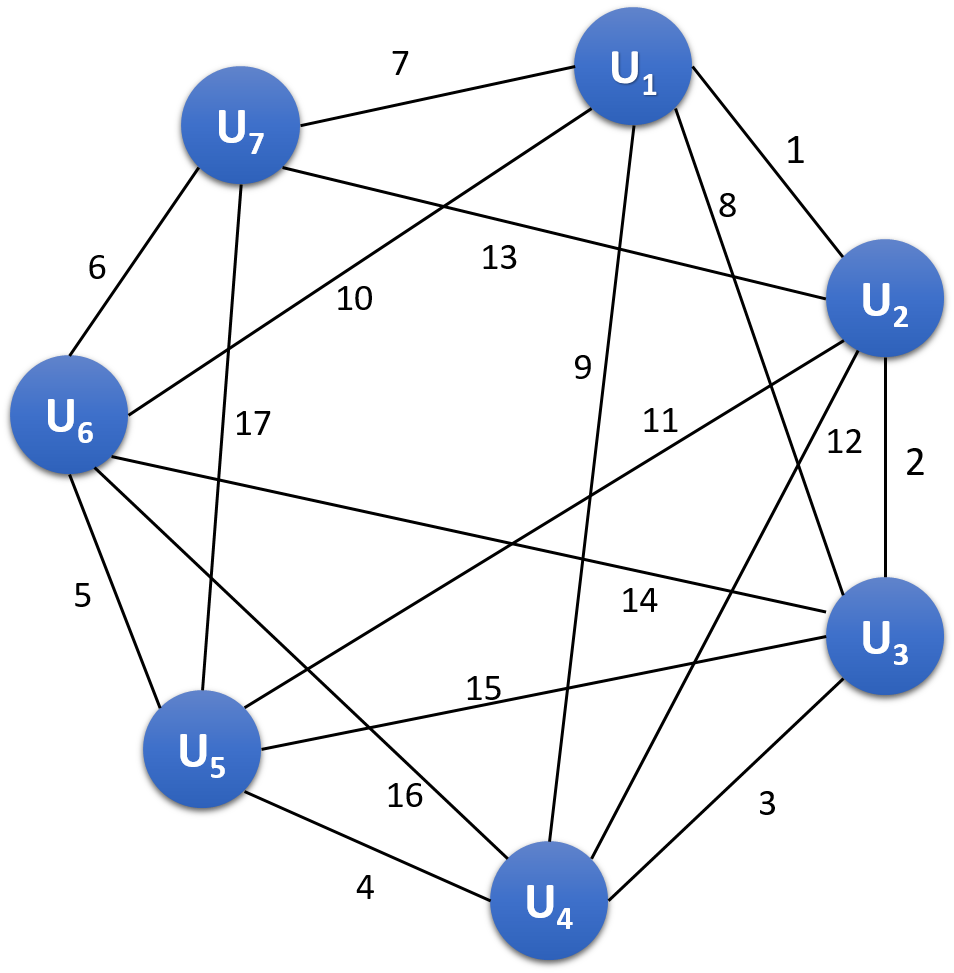}
    \caption{Partial Regular graph PRG($n = 7$ ,$d = 5$). Here all vertices (except $U_7$) have degree 5. Each vertex represents a node and edges represent vectors corresponding to the common packet between the nodes.}
    \label{1}
\end{figure}

  \begin{figure}
      \centering
       \includegraphics[scale=0.3]{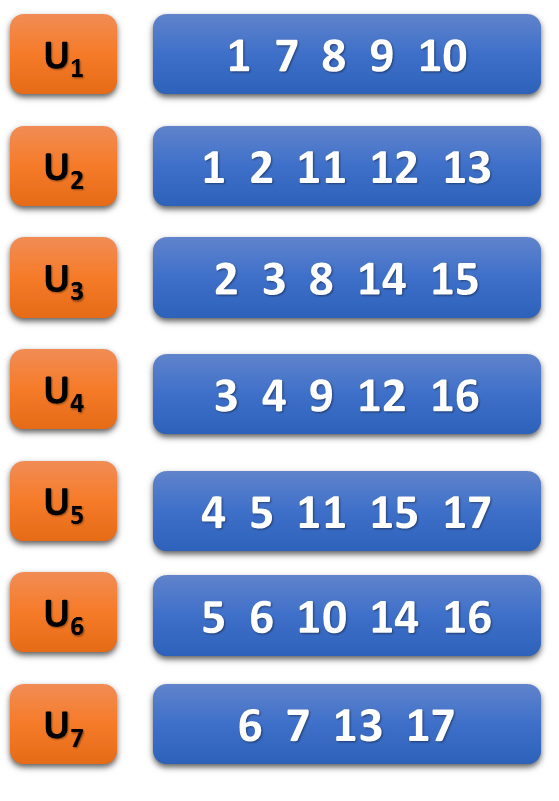}
  \captionof{figure} {Weak Fractional Repetition code for PRG(7,5) given in figure-\ref{1}. Node capacity of each node except $U_7$ is $\alpha=5$ and node $U_7$ has capacity $\alpha-1 = 4$.}
  \label{2}
  \end{figure}
 
\section{Ring construction for $\rho \geq 2$}
In this section we generalize the ring construction described in \cite{DBLP:journals/corr/abs-1302-3681} for $\rho \geq 2$. So if we have $\theta$ number of packets to be distributed among $n$ nodes such that each packet is replicated $\rho$ times then we place $n$ nodes on circular ring and then start placing packets starting from first node in between the successive nodes for $\rho-1$ time. So first packet will be placed between $\rho-1$ pairs of successive nodes. Here $\theta = qn + r , 0 \leq r \leq n-1$ If $r=0$ then this construction will give FR codes else it'll give WFR codes. We also give the details about the parameters which give \textit{universally good codes}. Example of FR code with $\rho = 3$ based on the ring construction is given below.

\begin{example}
Figure \ref{4} is an example of ring construction for distributing $20$ packets among $9$ nodes with replication factor $\rho=3$. Place packet $1$ between node $U_0,U_1$ and $U_1,U_2$. Start placing other packets on between the nodes in this manner. Figure \ref{5} is weak FR code based on this ring with parameters $n = 9, \rho = 3, \theta = 20$.
\end{example}

\begin{theorem}
Reconstruction degree of FR code where $\rho \geq 2$ for homogeneous DSSs  constructed using ring construction is as follows:

\[
k = \begin{cases}
n-\rho  &\text{if $n = \theta$} \\
n-\rho+1 &\text{if $n = m\theta, m>1, m \in N$}\\
 \end{cases}
\]

\label{th-3}
\end{theorem}

\begin{proof}
Observe that the node packet incidence matrix of any FR code of ring construction is either circulant or block circulant.
For $n = \theta$, all diagonal elements in the node packet incidence matrix (for FR code $(n=5, k=3, d=2)$, where $\rho = 2$ and $\theta = 5$) as given in Equation (\ref{n=t}) is 1. So by contacting any $n-\rho$ nodes, one can get distinct $n-\rho$ packets(which are diagonal elements). Since $n = \theta$ each node will have exactly $\rho$ distinct packets. thus there will be always one extra packet which is not on diagonal. So by contacting any $n-\rho$ node one can get $\theta-1$ distinct packets to reconstruct the entire file.


For $n = m\theta$, the node packet incidence matrix (for FR code $(n=5, k=3, d=2)$, where $\rho = 2$ and $\theta = 10$) can be divided in $m$ blocks of $n \times \theta$ size as shown in equation (\ref{n=2t}). All the blocks are exactly same as for $n = \theta$ matrix. So using the above proof we can say that if we contact any $n-\rho$ nodes we can get $(\theta/m)-1$ packet from each block. We have $m$ blocks,so by contacting $n-\rho$ nodes we can get $\theta - m$ distinct packets. Contacting $1$ extra node will get $m$ more packets(diagonal packet). So, by contacting any $n- \rho +1$ nodes one can get at least $m\theta-1$ packet out of $m\theta$ packet which is required to reconstruct the entire file.

\begin{equation}
M_{5\times 5} =
\begin{bmatrix}
1 & 0 & 0 & 0 & 1 \\
1 & 1 & 0 & 0 & 0 \\
0 & 1 & 1 & 0 & 0 \\
0 & 0 & 1 & 1 & 0 \\
0 & 0 & 0 & 1 & 1 \\
\end{bmatrix}_{5\times 5}\\
\label{n=t}
\end{equation}

\begin{equation}
M_{5\times 10} =
\begin{bmatrix}
1 & 0 & 0 & 0 & 1 & 1 & 0 & 0 & 0 & 1\\
1 & 1 & 0 & 0 & 0 & 1 & 1 & 0 & 0 & 0\\
0 & 1 & 1 & 0 & 0 & 0 & 1 & 1 & 0 & 0 \\
0 & 0 & 1 & 1 & 0 & 0 & 0 & 1 & 1 & 0\\
0 & 0 & 0 & 1 & 1 & 0 & 0 & 0 & 1 & 1\\
\end{bmatrix}_{5\times 10} \\
\label{n=2t}
\end{equation}
\end{proof}

Reconstruction degree of FR code for homogeneous DSSs is discussed, now one has to check whether constructed FR codes are universally good or not.For universally good codes, one has to prove that constructed codes with $[\theta,\theta-1]$ MDS inner codes and $(n, k, d)$ inner FR codes satisfy the equation-(\ref{eq1}).\\
\textbf{Case-1: $n = \theta$}:\\
As discussed in Theorem-\ref{th-3} for $n = \theta$, reconstruction degree $k$ = $n - \rho$. Now $M(k) = \theta-1$, $k = \theta-\rho$ and (using Equation-(\ref{eq_4})) $d = \alpha = \rho$. After substituting the as above in Equation-(\ref{eq1}), we can get the following inequality.
\[
3\rho^2+\theta^2-4\rho\theta+\theta+\rho-2 \geq 0.
\]
The solution of the inequality for $\rho$ and $\theta$ will give the desired result. Some values for are given in Tables \ref{t_r4},\ref{t_r3},\ref{t_r2}.
\textbf{Case-2: $n = m\theta$}:\\
As discussed in Theorem-\ref{th-3} for $n = m\theta$, $m > 1, m \in N$, reconstruction degree $k$ = $n-\rho+1$. Now $M(k) = \theta-1$, $k = m\theta-\rho+1$ and (using Equation-(\ref{eq_4})) $d = \alpha = \dfrac{\rho}{m}$. After substituting the as above in Equation-(\ref{eq1}), we can get the following inequality.
\[
\begin{array}{c}
  m^3 \theta^2+(m+2)\rho^2-2m^2 \theta\rho+m^2 \theta\\
 -(m+2)\rho-2m\theta \rho+2m\theta-2m \geq 0. 
\end{array}
\]

The solution of the inequality for $\rho$, $m$ and $\theta$ will give the desired result. Some values for are given in Tables \ref{t_r4},\ref{t_r3},\ref{t_r2}.

We discussed the reconstruction degree of FR codes for homogeneous DSSs. Now, we give a conjecture for heterogeneous DSSs.

\begin{theorem}
The reconstruction degree of FR code for heterogeneous DSSs constructed using ring construction is as follows.
\[
k = \begin{cases}
n-\rho & \text{if $n > \theta$} \\
n-\rho+1 & \text{if $n < \theta \neq m\theta, m \in N$}
\end{cases}
\]
\end{theorem}

\begin{table}
	\caption{Parameters for universally good FR codes using ring construction for $\rho = 4$}
	\centering 
	\begin{tabular}{|c|c|c|c|}
	\hline
	 n & k & d & $\theta$ \\ 
	 \hline
 10 & 6 & 4 & 10 \\
\hline
11 & 7 & 4 & 11 \\
\hline
12 & 8 & 4 & 12 \\
\hline
13 & 9 & 4 & 13 \\
\hline
14 & 10 & 4 & 14 \\
\hline
15 & 11 & 4 & 15 \\
\hline
16 & 12 & 4 & 16 \\
\hline

\end{tabular}
 \label{t_r4}
\end{table}

\begin{table}
	\caption{Parameters for universally good FR codes using ring construction for $\rho = 3$}
	\centering 
	\begin{tabular}{|c|c|c|c|}
   \hline
    n & k & d & $\theta$ \\ 
    \hline
7 & 4 & 3 & 7 \\
\hline
8 & 5 & 3 & 8 \\
\hline
9 & 6 & 3 & 9 \\
\hline
10 & 7 & 3 & 10 \\
\hline
11 & 8 & 3 & 11 \\
\hline
11 & 9 & 6 & 22 \\
\hline
12 & 9 & 3 & 12 \\
\hline
12 & 10 & 6 & 24 \\
\hline
13 & 10 & 3 & 13 \\
\hline
13 & 11 & 6 & 26 \\
\hline
14 & 11 & 3 & 14 \\
\hline
14 & 12 & 6 & 28 \\
\hline
15 & 12 & 3 & 15 \\
\hline
16 & 13 & 3 & 16 \\
\hline

\end{tabular}
 \label{t_r3}
\end{table}

\begin{table}
	\caption{Parameters for universally good FR codes using ring construction for $\rho = 2$}
	\centering 
	\begin{tabular}{|c|c|c|c|}
   \hline
    n & k & d & $\theta$ \\ 
    \hline
   4 & 2 & 2 & 4 \\
\hline
5 & 3 & 2 & 5 \\
\hline
6 & 4 & 2 & 6 \\
\hline
6 & 5 & 4 & 12 \\
\hline
7 & 5 & 2 & 7 \\
\hline
7 & 6 & 4 & 14 \\
\hline
8 & 6 & 2 & 8 \\
\hline
8 & 7 & 4 & 16 \\
\hline
8 & 7 & 6 & 24 \\
\hline
9 & 7 & 2 & 9 \\
\hline
9 & 8 & 4 & 18 \\
\hline
9 & 8 & 6 & 27 \\
\hline
10 & 8 & 2 & 10 \\
\hline
10 & 9 & 4 & 20 \\
\hline
10 & 9 & 6 & 30 \\
\hline
11 & 9 & 2 & 11 \\
\hline
11 & 10 & 4 & 22 \\
\hline
11 & 10 & 6 & 33 \\
\hline
12 & 10 & 2 & 12 \\
\hline
12 & 11 & 4 & 24 \\
\hline
13 & 11 & 2 & 13 \\
\hline
13 & 12 & 4 & 26 \\
\hline
14 & 12 & 2 & 14 \\
\hline
15 & 13 & 2 & 15 \\
\hline

\end{tabular}
 \label{t_r2}
\end{table}

\begin{figure}
    \centering
    \includegraphics[scale=0.6]{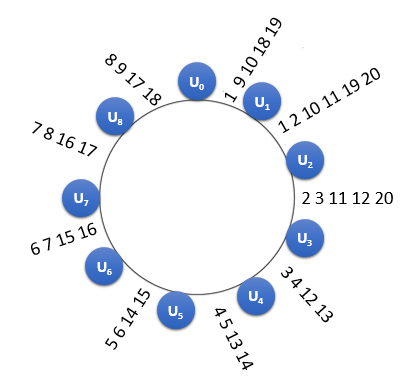}
    \caption{Generalized ring construction for $\theta = 20$, $\rho=3$ and $n=9$. Each packet is placed between consecutive nodes for $rho-1$ times.}
    \label{4}
\end{figure}
\begin{figure}
    \centering
    \includegraphics[scale=0.3]{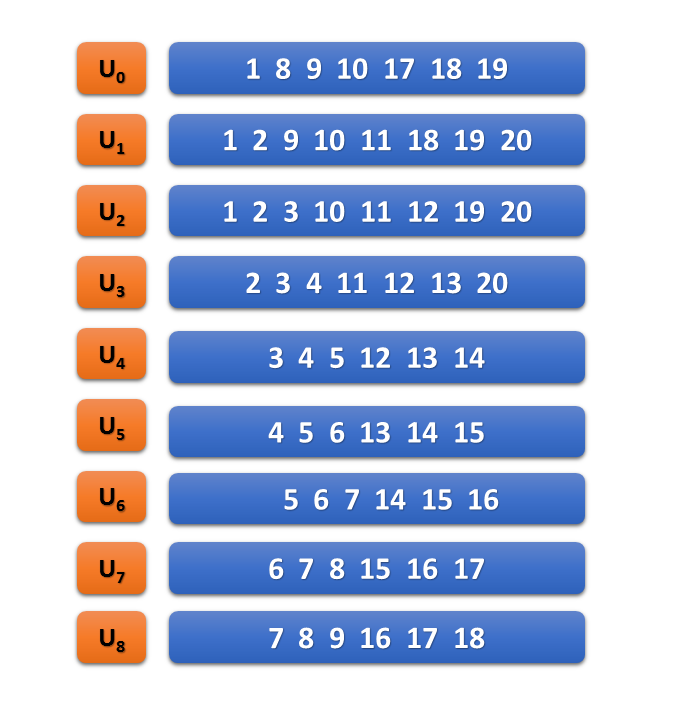}
    \caption{Weak Fractional Repetition code based on generalized ring construction for $\theta = 20$, $\rho = 3$ and $n = 9$.}
    \label{5}
\end{figure}
\section{T-construction results}
In Proposition 2 of \cite{DBLP:journals/corr/abs-1302-3681}, researchers have discussed about the unique approach for the construction of homogeneous FR codes where, $\rho = d$ and $n = \theta$. In this section we give parameters of FR codes constructed using this method which are universally good codes. Here, we observe that for different $t$ values constructed FR codes have same parameters. In Table-\ref{t13} one can observe that for $t=0$ and $t=2$ constructed FR codes have same parameters $(n=4, k=2, d=2, \theta=4)$.

\begin{table}
	\caption{Parameters for universally good FR codes using t-construction}
	\centering 
	\begin{tabular}{|c|c|c|c|c|c|}
   \hline
    n & k & d & $\rho$ & $\theta$ & t \\ 
    \hline
  4 & 2 & 2 & 2 & 4 & 0 \\
\hline
4 & 3 & 2 & 2 & 4 & 1 \\
\hline
4 & 2 & 2 & 2 & 4 & 2 \\
\hline
4 & 3 & 2 & 2 & 4 & 3 \\
\hline
4 & 2 & 2 & 2 & 4 & 4 \\
\hline
5 & 3 & 2 & 2 & 5 & 0 \\
\hline
5 & 4 & 2 & 2 & 5 & 1 \\
\hline
5 & 3 & 2 & 2 & 5 & 2 \\
\hline
5 & 3 & 2 & 2 & 5 & 3 \\
\hline
5 & 3 & 2 & 2 & 5 & 4 \\
\hline
6 & 4 & 2 & 2 & 6 & 0 \\
\hline
6 & 5 & 2 & 2 & 6 & 1 \\
\hline
6 & 4 & 2 & 2 & 6 & 2 \\
\hline
6 & 4 & 2 & 2 & 6 & 3 \\
\hline
6 & 5 & 2 & 2 & 6 & 4 \\
\hline
6 & 5 & 3 & 3 & 6 & 1 \\
\hline
6 & 5 & 3 & 3 & 6 & 4 \\
\hline
7 & 5 & 2 & 2 & 7 & 0 \\
\hline
7 & 6 & 2 & 2 & 7 & 1 \\
\hline
7 & 5 & 2 & 2 & 7 & 2 \\
\hline
7 & 5 & 2 & 2 & 7 & 3 \\
\hline
7 & 5 & 2 & 2 & 7 & 4 \\
\hline
7 & 5 & 3 & 3 & 7 & 0 \\
\hline
7 & 6 & 3 & 3 & 7 & 1 \\
\hline
7 & 3 & 3 & 3 & 7 & 2 \\
\hline
7 & 4 & 3 & 3 & 7 & 3 \\
\hline
7 & 3 & 3 & 3 & 7 & 4 \\
\hline
8 & 6 & 2 & 2 & 8 & 0 \\
\hline
8 & 7 & 2 & 2 & 8 & 1 \\
\hline
8 & 6 & 2 & 2 & 8 & 2 \\
\hline
8 & 6 & 2 & 2 & 8 & 3 \\
\hline
8 & 6 & 2 & 2 & 8 & 4 \\
\hline
8 & 6 & 3 & 3 & 8 & 0 \\
\hline
8 & 7 & 3 & 3 & 8 & 1 \\
\hline
8 & 4 & 3 & 3 & 8 & 2 \\
\hline
8 & 6 & 3 & 3 & 8 & 3 \\
\hline
8 & 5 & 3 & 3 & 8 & 4 \\
\hline
8 & 7 & 4 & 4 & 8 & 1 \\
\hline
9 & 7 & 2 & 2 & 9 & 0 \\
\hline
9 & 8 & 2 & 2 & 9 & 1 \\
\hline
9 & 7 & 2 & 2 & 9 & 2 \\
\hline
9 & 7 & 2 & 2 & 9 & 3 \\
\hline
9 & 7 & 2 & 2 & 9 & 4 \\
\hline
9 & 7 & 3 & 3 & 9 & 0 \\
\hline
9 & 8 & 3 & 3 & 9 & 1 \\
\hline
9 & 5 & 3 & 3 & 9 & 2 \\
\hline
9 & 5 & 3 & 3 & 9 & 3 \\
\hline
9 & 7 & 3 & 3 & 9 & 4 \\
\hline
9 & 7 & 4 & 4 & 9 & 0 \\
\hline
9 & 8 & 4 & 4 & 9 & 1 \\
\hline
9 & 7 & 4 & 4 & 9 & 4 \\
\hline
10 & 8 & 2 & 2 & 10 & 0 \\
\hline
10 & 9 & 2 & 2 & 10 & 1 \\
\hline
10 & 8 & 2 & 2 & 10 & 2 \\
\hline
10 & 8 & 2 & 2 & 10 & 3 \\
\hline
10 & 8 & 2 & 2 & 10 & 4 \\
\hline
10 & 8 & 3 & 3 & 10 & 0 \\
\hline
10 & 9 & 3 & 3 & 10 & 1 \\
\hline
10 & 6 & 3 & 3 & 10 & 2 \\
\hline
10 & 7 & 3 & 3 & 10 & 3 \\
\hline
10 & 7 & 3 & 3 & 10 & 4 \\
\hline
10 & 8 & 4 & 4 & 10 & 0 \\
\hline
10 & 9 & 4 & 4 & 10 & 1 \\
\hline
10 & 6 & 4 & 4 & 10 & 3 \\
\hline
10 & 7 & 4 & 4 & 10 & 4 \\
\hline
11 & 9 & 2 & 2 & 11 & 0 \\
\hline
11 & 10 & 2 & 2 & 11 & 1 \\
\hline
11 & 9 & 2 & 2 & 11 & 2 \\
\hline
11 & 9 & 2 & 2 & 11 & 3 \\
\hline
11 & 9 & 2 & 2 & 11 & 4 \\
\hline
11 & 9 & 3 & 3 & 11 & 0 \\
\hline
11 & 10 & 3 & 3 & 11 & 1 \\
\hline
11 & 7 & 3 & 3 & 11 & 2 \\
\hline
11 & 7 & 3 & 3 & 11 & 3 \\
\hline
11 & 7 & 3 & 3 & 11 & 4 \\
\hline
11 & 9 & 4 & 4 & 11 & 0 \\
\hline
11 & 10 & 4 & 4 & 11 & 1 \\
\hline
11 & 6 & 4 & 4 & 11 & 2 \\
\hline
11 & 6 & 4 & 4 & 11 & 3 \\
\hline
11 & 6 & 4 & 4 & 11 & 4 \\
\hline

\end{tabular}
 \label{t1}
\end{table}

\begin{table}
	\caption{Parameters for universally good FR codes using t-construction}
	\centering 
	\begin{tabular}{|c|c|c|c|c|c|}
   \hline
    n & k & d & $\rho$ & $\theta$ & t \\ 
    \hline
 12 & 10 & 2 & 2 & 12 & 0 \\
\hline
12 & 11 & 2 & 2 & 12 & 1 \\
\hline
12 & 10 & 2 & 2 & 12 & 2 \\
\hline
12 & 10 & 2 & 2 & 12 & 3 \\
\hline
12 & 10 & 2 & 2 & 12 & 4 \\
\hline
12 & 10 & 3 & 3 & 12 & 0 \\
\hline
12 & 11 & 3 & 3 & 12 & 1 \\
\hline
12 & 8 & 3 & 3 & 12 & 2 \\
\hline
12 & 9 & 3 & 3 & 12 & 3 \\
\hline
12 & 10 & 3 & 3 & 12 & 4 \\
\hline
12 & 10 & 4 & 4 & 12 & 0 \\
\hline
12 & 11 & 4 & 4 & 12 & 1 \\
\hline
12 & 7 & 4 & 4 & 12 & 2 \\
\hline
12 & 9 & 4 & 4 & 12 & 3 \\
\hline
12 & 10 & 4 & 4 & 12 & 4 \\
\hline
13 & 11 & 2 & 2 & 13 & 0 \\
\hline
13 & 12 & 2 & 2 & 13 & 1 \\
\hline
13 & 11 & 2 & 2 & 13 & 2 \\
\hline
13 & 11 & 2 & 2 & 13 & 3 \\
\hline
13 & 11 & 2 & 2 & 13 & 4 \\
\hline
13 & 11 & 3 & 3 & 13 & 0 \\
\hline
13 & 12 & 3 & 3 & 13 & 1 \\
\hline
13 & 9 & 3 & 3 & 13 & 2 \\
\hline
13 & 9 & 3 & 3 & 13 & 3 \\
\hline
13 & 9 & 3 & 3 & 13 & 4 \\
\hline
13 & 11 & 4 & 4 & 13 & 0 \\
\hline
13 & 12 & 4 & 4 & 13 & 1 \\
\hline
13 & 8 & 4 & 4 & 13 & 2 \\
\hline
13 & 9 & 4 & 4 & 13 & 3 \\
\hline
13 & 8 & 4 & 4 & 13 & 4 \\
\hline
14 & 12 & 2 & 2 & 14 & 0 \\
\hline
14 & 13 & 2 & 2 & 14 & 1 \\
\hline
14 & 12 & 2 & 2 & 14 & 2 \\
\hline
14 & 12 & 2 & 2 & 14 & 3 \\
\hline
14 & 12 & 2 & 2 & 14 & 4 \\
\hline
14 & 12 & 3 & 3 & 14 & 0 \\
\hline
14 & 10 & 3 & 3 & 14 & 2 \\
\hline
14 & 11 & 3 & 3 & 14 & 3 \\
\hline
14 & 10 & 3 & 3 & 14 & 4 \\
\hline
14 & 12 & 4 & 4 & 14 & 0 \\
\hline
14 & 9 & 4 & 4 & 14 & 2 \\
\hline
14 & 9 & 4 & 4 & 14 & 3 \\
\hline
14 & 9 & 4 & 4 & 14 & 4 \\
\hline
15 & 13 & 2 & 2 & 15 & 0 \\
\hline
15 & 13 & 2 & 2 & 15 & 2 \\
\hline
15 & 13 & 2 & 2 & 15 & 3 \\
\hline
15 & 13 & 2 & 2 & 15 & 4 \\
\hline
15 & 11 & 3 & 3 & 15 & 2 \\
\hline
15 & 11 & 3 & 3 & 15 & 3 \\
\hline
15 & 9 & 4 & 4 & 15 & 2 \\
\hline
15 & 10 & 4 & 4 & 15 & 3 \\
\hline
16 & 12 & 3 & 3 & 16 & 2 \\
\hline
16 & 13 & 3 & 3 & 16 & 3 \\
\hline
16 & 12 & 3 & 3 & 16 & 4 \\
\hline
16 & 10 & 4 & 4 & 16 & 2 \\
\hline
16 & 12 & 4 & 4 & 16 & 4 \\
\hline
17 & 13 & 3 & 3 & 17 & 2 \\
\hline
17 & 13 & 3 & 3 & 17 & 3 \\
\hline
17 & 13 & 3 & 3 & 17 & 4 \\
\hline
17 & 11 & 4 & 4 & 17 & 2 \\
\hline
17 & 12 & 4 & 4 & 17 & 3 \\
\hline
17 & 12 & 4 & 4 & 17 & 4 \\
\hline
18 & 12 & 4 & 4 & 18 & 2 \\
\hline

\end{tabular}
 \label{t2}
\end{table}

\begin{table}
	\caption{Parameters for universally good FR codes using t-construction for $RHS > 0$ in Equation (\ref{eq1})}
	\centering 
	\begin{tabular}{|c|c|c|c|c|c|}
   \hline
    n & k & d & $\rho$ & $\theta$ & t \\ 
    \hline
 4 & 2 & 2 & 2 & 4 & 0 \\
\hline
4 & 3 & 2 & 2 & 4 & 1 \\
\hline
4 & 2 & 2 & 2 & 4 & 2 \\
\hline
4 & 3 & 2 & 2 & 4 & 3 \\
\hline
4 & 2 & 2 & 2 & 4 & 4 \\
\hline
5 & 3 & 2 & 2 & 5 & 0 \\
\hline
5 & 4 & 2 & 2 & 5 & 1 \\
\hline
5 & 3 & 2 & 2 & 5 & 2 \\
\hline
5 & 3 & 2 & 2 & 5 & 3 \\
\hline
5 & 3 & 2 & 2 & 5 & 4 \\
\hline
6 & 4 & 2 & 2 & 6 & 0 \\
\hline
6 & 5 & 2 & 2 & 6 & 1 \\
\hline
6 & 4 & 2 & 2 & 6 & 2 \\
\hline
6 & 4 & 2 & 2 & 6 & 3 \\
\hline
6 & 5 & 2 & 2 & 6 & 4 \\
\hline
6 & 5 & 3 & 3 & 6 & 1 \\
\hline
6 & 5 & 3 & 3 & 6 & 4 \\
\hline
7 & 5 & 2 & 2 & 7 & 0 \\
\hline
7 & 5 & 2 & 2 & 7 & 2 \\
\hline
7 & 5 & 2 & 2 & 7 & 3 \\
\hline
7 & 5 & 2 & 2 & 7 & 4 \\
\hline
7 & 5 & 3 & 3 & 7 & 0 \\
\hline
7 & 6 & 3 & 3 & 7 & 1 \\
\hline
7 & 3 & 3 & 3 & 7 & 2 \\
\hline
7 & 4 & 3 & 3 & 7 & 3 \\
\hline
7 & 3 & 3 & 3 & 7 & 4 \\
\hline
8 & 6 & 3 & 3 & 8 & 0 \\
\hline
8 & 7 & 3 & 3 & 8 & 1 \\
\hline
8 & 4 & 3 & 3 & 8 & 2 \\
\hline
8 & 6 & 3 & 3 & 8 & 3 \\
\hline
8 & 5 & 3 & 3 & 8 & 4 \\
\hline
8 & 7 & 4 & 4 & 8 & 1 \\
\hline
9 & 7 & 3 & 3 & 9 & 0 \\
\hline
9 & 5 & 3 & 3 & 9 & 2 \\
\hline
9 & 5 & 3 & 3 & 9 & 3 \\
\hline
9 & 7 & 3 & 3 & 9 & 4 \\
\hline
9 & 7 & 4 & 4 & 9 & 0 \\
\hline
9 & 8 & 4 & 4 & 9 & 1 \\
\hline
9 & 7 & 4 & 4 & 9 & 4 \\
\hline
10 & 6 & 3 & 3 & 10 & 2 \\
\hline
10 & 7 & 3 & 3 & 10 & 3 \\
\hline
10 & 7 & 3 & 3 & 10 & 4 \\
\hline
10 & 8 & 4 & 4 & 10 & 0 \\
\hline
10 & 9 & 4 & 4 & 10 & 1 \\
\hline
10 & 6 & 4 & 4 & 10 & 3 \\
\hline
10 & 7 & 4 & 4 & 10 & 4 \\
\hline
11 & 7 & 3 & 3 & 11 & 2 \\
\hline
11 & 7 & 3 & 3 & 11 & 3 \\
\hline
11 & 7 & 3 & 3 & 11 & 4 \\
\hline
11 & 9 & 4 & 4 & 11 & 0 \\
\hline
11 & 6 & 4 & 4 & 11 & 2 \\
\hline
11 & 6 & 4 & 4 & 11 & 3 \\
\hline
11 & 6 & 4 & 4 & 11 & 4 \\
\hline
12 & 7 & 4 & 4 & 12 & 2 \\
\hline
12 & 9 & 4 & 4 & 12 & 3 \\
\hline
13 & 8 & 4 & 4 & 13 & 2 \\
\hline
13 & 9 & 4 & 4 & 13 & 3 \\
\hline
13 & 8 & 4 & 4 & 13 & 4 \\
\hline
14 & 13 & 2 & 2 & 14 & 1 \\
\hline
14 & 9 & 4 & 4 & 14 & 2 \\
\hline
14 & 9 & 4 & 4 & 14 & 3 \\
\hline
14 & 9 & 4 & 4 & 14 & 4 \\
\hline
15 & 13 & 2 & 2 & 15 & 0 \\
\hline
15 & 13 & 2 & 2 & 15 & 2 \\
\hline
15 & 13 & 2 & 2 & 15 & 3 \\
\hline
15 & 13 & 2 & 2 & 15 & 4 \\
\hline
15 & 9 & 4 & 4 & 15 & 2 \\
\hline
16 & 13 & 3 & 3 & 16 & 3 \\
\hline
17 & 13 & 3 & 3 & 17 & 2 \\
\hline
17 & 13 & 3 & 3 & 17 & 3 \\
\hline
17 & 13 & 3 & 3 & 17 & 4 \\
\hline
\end{tabular}
\label{t13}
\end{table}

\begin{table}
	\caption{Parameters for universally good FR codes using t-construction for $RHS > 0$ in Equation (\ref{eq1}) and no repetition of codes with same parameters}
	\centering 
	\begin{tabular}{|c|c|c|c|c|c|}
   \hline
    n & k & d & $\rho$ & $\theta$ & t \\ 
    \hline
4 & 2 & 2 & 2 & 4 & 0 \\
\hline
4 & 3 & 2 & 2 & 4 & 1 \\
\hline
5 & 3 & 2 & 2 & 5 & 0 \\
\hline
5 & 4 & 2 & 2 & 5 & 1 \\
\hline
6 & 4 & 2 & 2 & 6 & 0 \\
\hline
6 & 5 & 2 & 2 & 6 & 1 \\
\hline
6 & 5 & 3 & 3 & 6 & 1 \\
\hline
7 & 5 & 2 & 2 & 7 & 0 \\
\hline
7 & 5 & 3 & 3 & 7 & 0 \\
\hline
7 & 6 & 3 & 3 & 7 & 1 \\
\hline
7 & 3 & 3 & 3 & 7 & 2 \\
\hline
7 & 4 & 3 & 3 & 7 & 3 \\
\hline
8 & 6 & 3 & 3 & 8 & 0 \\
\hline
8 & 7 & 3 & 3 & 8 & 1 \\
\hline
8 & 4 & 3 & 3 & 8 & 2 \\
\hline
8 & 5 & 3 & 3 & 8 & 4 \\
\hline
8 & 7 & 4 & 4 & 8 & 1 \\
\hline
9 & 7 & 3 & 3 & 9 & 0 \\
\hline
9 & 5 & 3 & 3 & 9 & 2 \\
\hline
9 & 7 & 4 & 4 & 9 & 0 \\
\hline
9 & 8 & 4 & 4 & 9 & 1 \\
\hline
10 & 6 & 3 & 3 & 10 & 2 \\
\hline
10 & 7 & 3 & 3 & 10 & 3 \\
\hline
10 & 8 & 4 & 4 & 10 & 0 \\
\hline
10 & 9 & 4 & 4 & 10 & 1 \\
\hline
10 & 6 & 4 & 4 & 10 & 3 \\
\hline
10 & 7 & 4 & 4 & 10 & 4 \\
\hline
11 & 7 & 3 & 3 & 11 & 2 \\
\hline
11 & 9 & 4 & 4 & 11 & 0 \\
\hline
11 & 6 & 4 & 4 & 11 & 2 \\
\hline
12 & 7 & 4 & 4 & 12 & 2 \\
\hline
12 & 9 & 4 & 4 & 12 & 3 \\
\hline
13 & 8 & 4 & 4 & 13 & 2 \\
\hline
13 & 9 & 4 & 4 & 13 & 3 \\
\hline
14 & 13 & 2 & 2 & 14 & 1 \\
\hline
14 & 9 & 4 & 4 & 14 & 2 \\
\hline
15 & 13 & 2 & 2 & 15 & 0 \\
\hline
15 & 9 & 4 & 4 & 15 & 2 \\
\hline
16 & 13 & 3 & 3 & 16 & 3 \\
\hline
17 & 13 & 3 & 3 & 17 & 2 \\
\hline

\end{tabular}
 \label{t12}
\end{table}

\begin{table}
	\caption{Parameters for universally good FR codes using t-construction for $RHS > 0$ in Equation (\ref{eq1}), no repetition of codes with same parameters and codes with $\rho = 2$}
	\centering 
	\begin{tabular}{|c|c|c|c|c|c|}
   \hline
    n & k & d & $\rho$ & $\theta$ & t \\ 
    \hline
4 & 2 & 2 & 2 & 4 & 0 \\
\hline
4 & 3 & 2 & 2 & 4 & 1 \\
\hline
5 & 3 & 2 & 2 & 5 & 0 \\
\hline
5 & 4 & 2 & 2 & 5 & 1 \\
\hline
6 & 4 & 2 & 2 & 6 & 0 \\
\hline
6 & 5 & 2 & 2 & 6 & 1 \\
\hline
7 & 5 & 2 & 2 & 7 & 0 \\
\hline
14 & 13 & 2 & 2 & 14 & 1 \\
\hline
15 & 13 & 2 & 2 & 15 & 0 \\
\hline
\end{tabular}
 \label{t11}
\end{table}


\begin{table}
	\caption{Parameters for universally good FR codes using t-construction for $RHS > 0$ in Equation (\ref{eq1}), no repetition of codes with same parameters and code with $\rho = 3$}
	\centering 
	\begin{tabular}{|c|c|c|c|c|c|}
   \hline
    n & k & d & $\rho$ & $\theta$ & t \\ 
    \hline
6 & 5 & 3 & 3 & 6 & 1 \\
\hline
7 & 5 & 3 & 3 & 7 & 0 \\
\hline
7 & 6 & 3 & 3 & 7 & 1 \\
\hline
7 & 3 & 3 & 3 & 7 & 2 \\
\hline
7 & 4 & 3 & 3 & 7 & 3 \\
\hline
8 & 6 & 3 & 3 & 8 & 0 \\
\hline
8 & 7 & 3 & 3 & 8 & 1 \\
\hline
8 & 4 & 3 & 3 & 8 & 2 \\
\hline
8 & 5 & 3 & 3 & 8 & 4 \\
\hline
9 & 7 & 3 & 3 & 9 & 0 \\
\hline
9 & 5 & 3 & 3 & 9 & 2 \\
\hline
10 & 6 & 3 & 3 & 10 & 2 \\
\hline
10 & 7 & 3 & 3 & 10 & 3 \\
\hline
11 & 7 & 3 & 3 & 11 & 2 \\
\hline
16 & 13 & 3 & 3 & 16 & 3 \\
\hline
17 & 13 & 3 & 3 & 17 & 2 \\
\hline
\end{tabular}
 \label{t10}
\end{table}

\section{Conclusion}
In this work, we show that the FR codes constructed using the Partial Regular graphs are universally good. We also give a new ring construction of FR codes. A conjecture for the reconstruction degree of FR codes (constructed using ring construction) is given. We identify the universally good FR codes for homogeneous DSSs in both ring construction and $t-$construction.
 \small
 \bibliographystyle{IEEEtran}
 \bibliography{cloud}

\begin{thebibliography}{10}
\providecommand{\url}[1]{#1}
\csname url@samestyle\endcsname
\providecommand{\newblock}{\relax}
\providecommand{\bibinfo}[2]{#2}
\providecommand{\BIBentrySTDinterwordspacing}{\spaceskip=0pt\relax}
\providecommand{\BIBentryALTinterwordstretchfactor}{4}
\providecommand{\BIBentryALTinterwordspacing}{\spaceskip=\fontdimen2\font plus
\BIBentryALTinterwordstretchfactor\fontdimen3\font minus
  \fontdimen4\font\relax}
\providecommand{\BIBforeignlanguage}[2]{{%
\expandafter\ifx\csname l@#1\endcsname\relax
\typeout{** WARNING: IEEEtran.bst: No hyphenation pattern has been}%
\typeout{** loaded for the language `#1'. Using the pattern for}%
\typeout{** the default language instead.}%
\else
\language=\csname l@#1\endcsname
\fi
#2}}
\providecommand{\BIBdecl}{\relax}
\BIBdecl

\bibitem{Ghemawat03}
\BIBentryALTinterwordspacing
S.~Ghemawat, H.~Gobioff, and S.-T. Leung, ``The google file system,'' in
  \emph{Proceedings of the nineteenth ACM symposium on Operating systems
  principles}, ser. SOSP '03.\hskip 1em plus 0.5em minus 0.4em\relax New York,
  NY, USA: ACM, 2003, pp. 29--43. [Online]. Available:
  \url{http://doi.acm.org/10.1145/945445.945450}
\BIBentrySTDinterwordspacing

\bibitem{5496972}
K.~Shvachko, H.~Kuang, S.~Radia, and R.~Chansler, ``The hadoop distributed file
  system,'' in \emph{2010 IEEE 26th Symposium on Mass Storage Systems and
  Technologies (MSST)}, May 2010, pp. 1--10.

\bibitem{XorbasVLDB}
M.~Sathiamoorthy, M.~Asteris, D.~Papailiopoulos, A.~G. Dimakis, R.~Vadali,
  S.~Chen, and D.~Borthakur, ``Xoring elephants: Novel erasure codes for big
  data,'' \emph{Proceedings of the VLDB Endowment (to appear)}, 2013.

\bibitem{huang2012erasure}
C.~Huang, H.~Simitci, Y.~Xu, A.~Ogus, B.~Calder, P.~Gopalan, J.~Li, and
  S.~Yekhanin, ``Erasure coding in windows azure storage,'' in \emph{Presented
  as part of the 2012 USENIX Annual Technical Conference (USENIX ATC 12)},
  2012, pp. 15--26.

\bibitem{5550492}
A.~Dimakis, P.~Godfrey, Y.~Wu, M.~Wainwright, and K.~Ramchandran, ``Network
  coding for distributed storage systems,'' \emph{Information Theory, IEEE
  Transactions on}, vol.~56, no.~9, pp. 4539--4551, Sept 2010.

\bibitem{6620424}
J.~Pernas, C.~Yuen, B.~Gaston, and J.~Pujol, ``Non-homogeneous two-rack model
  for distributed storage systems,'' in \emph{Information Theory Proceedings
  (ISIT), 2013 IEEE International Symposium on}, July 2013, pp. 1237--1241.

\bibitem{DBLP:journals/corr/BenerjeeG15}
K.~G. Benerjee and M.~K. Gupta, ``Tradeoff for heterogeneous distributed
  storage systems between storage and repair cost,'' \emph{CoRR}, vol.
  abs/1503.02276, 2015.

\bibitem{ETT:ETT2887}
\BIBentryALTinterwordspacing
Q.~Yu, K.~W. Shum, and C.~W. Sung, ``Tradeoff between storage cost and repair
  cost in heterogeneous distributed storage systems,'' \emph{Transactions on
  Emerging Telecommunications Technologies}, pp. n/a--n/a, 2014. [Online].
  Available: \url{http://dx.doi.org/10.1002/ett.2887}
\BIBentrySTDinterwordspacing

\bibitem{Akhlaghi20102105}
\BIBentryALTinterwordspacing
S.~Akhlaghi, A.~Kiani, and M.~R. Ghanavati, ``Cost-bandwidth tradeoff in
  distributed storage systems,'' \emph{Computer Communications}, vol.~33,
  no.~17, pp. 2105 -- 2115, 2010, special Issue:Applied sciences in
  communication technologies. [Online]. Available:
  \url{http://www.sciencedirect.com/science/article/pii/S0140366410003506}
\BIBentrySTDinterwordspacing

\bibitem{RSKR09}
K.~Rashmi, N.~Shah, P.~Kumar, and K.~Ramchandran, ``Explicit construction of
  optimal exact regenerating codes for distributed storage,'' in
  \emph{Communication, Control, and Computing, 2009. Allerton 2009. 47th Annual
  Allerton Conference on}, 30 2009-oct. 2 2009, pp. 1243 --1249.

\bibitem{rr10}
S.~El~Rouayheb and K.~Ramchandran, ``Fractional repetition codes for repair in
  distributed storage systems,'' in \emph{Communication, Control, and Computing
  (Allerton), 2010 48th Annual Allerton Conference on}, Oct. 2010, pp. 1510
  --1517.

\bibitem{DBLP:journals/corr/abs-1201-3547}
T.~Ernvall, ``The existence of fractional repetition codes,'' \emph{CoRR}, vol.
  abs/1201.3547, 2012.

\bibitem{6120326}
J.~Koo and J.~Gill, ``Scalable constructions of fractional repetition codes in
  distributed storage systems,'' in \emph{Communication, Control, and Computing
  (Allerton), 2011 49th Annual Allerton Conference on}, Sept 2011, pp.
  1366--1373.

\bibitem{6483351}
O.~Olmez and A.~Ramamoorthy, ``Repairable replication-based storage systems
  using resolvable designs,'' in \emph{Communication, Control, and Computing
  (Allerton), 2012 50th Annual Allerton Conference on}, Oct 2012, pp.
  1174--1181.

\bibitem{6810361}
------, ``Constructions of fractional repetition codes from combinatorial
  designs,'' in \emph{Signals, Systems and Computers, 2013 Asilomar Conference
  on}, Nov 2013, pp. 647--651.

\bibitem{6033980}
S.~Pawar, N.~Noorshams, S.~El~Rouayheb, and K.~Ramchandran, ``Dress codes for
  the storage cloud: Simple randomized constructions,'' in \emph{Information
  Theory Proceedings (ISIT), 2011 IEEE International Symposium on}, July 2011,
  pp. 2338--2342.

\bibitem{DBLP:journals/corr/abs-1303-6801}
S.~Anil, M.~K. Gupta, and T.~A. Gulliver, ``Enumerating some fractional
  repetition codes,'' \emph{CoRR}, vol. abs/1303.6801, 2013.

\bibitem{6804948}
Q.~Yu, C.~W. Sung, and T.~Chan, ``Irregular fractional repetition code
  optimization for heterogeneous cloud storage,'' \emph{Selected Areas in
  Communications, IEEE Journal on}, vol.~32, no.~5, pp. 1048--1060, May 2014.

\bibitem{6763122}
B.~Zhu, K.~Shum, H.~Li, and H.~Hou, ``General fractional repetition codes for
  distributed storage systems,'' \emph{Communications Letters, IEEE}, vol.~18,
  no.~4, pp. 660--663, April 2014.

\bibitem{7118709}
N.~Silberstein and T.~Etzion, ``Optimal fractional repetition codes based on
  graphs and designs,'' \emph{Information Theory, IEEE Transactions on},
  vol.~61, no.~8, pp. 4164--4180, Aug 2015.

\bibitem{iet:/content/journals/10.1049/iet-com.2014.1225}
\BIBentryALTinterwordspacing
B.~Zhu, H.~Li, K.~W. Shum, and S.-Y.~R. Li,
  ``\BIBforeignlanguage{English}{{HFR} code: a flexible replication scheme for
  cloud storage systems},'' \emph{\BIBforeignlanguage{English}{IET
  Communications}}, October 2015. [Online]. Available:
  \url{http://digital-library.theiet.org/content/journals/10.1049/iet-com.2014.1225}
\BIBentrySTDinterwordspacing

\bibitem{DBLP:journals/corr/abs-1302-3681}
\BIBentryALTinterwordspacing
M.~K. Gupta, A.~Agrawal, and D.~Yadav, ``On weak dress codes for cloud
  storage,'' \emph{CoRR}, vol. abs/1302.3681, 2013. [Online]. Available:
  \url{http://arxiv.org/abs/1302.3681}
\BIBentrySTDinterwordspacing

\bibitem{KGBenerjee}
K.~G. Benerjee and M.~K. Gupta, ``On dress codes with flowers,'' \emph{Signal
  Design and Its Applications in Communications, The Seventh International
  Workshop on}, pp. 108--112, Sept. 2015.

\bibitem{5707092}
S.~E. Rouayheb and K.~Ramchandran, ``Fractional repetition codes for repair in
  distributed storage systems,'' in \emph{Communication, Control, and Computing
  (Allerton), 2010 48th Annual Allerton Conference on}, Sept 2010, pp.
  1510--1517.

\bibitem{Rashmi:2009:ECO:1793974.1794188}
\BIBentryALTinterwordspacing
K.~V. Rashmi, N.~B. Shah, P.~V. Kumar, and K.~Ramchandran, ``Explicit
  construction of optimal exact regenerating codes for distributed storage,''
  in \emph{Proceedings of the 47th Annual Allerton Conference on Communication,
  Control, and Computing}, ser. Allerton'09.\hskip 1em plus 0.5em minus
  0.4em\relax Piscataway, NJ, USA: IEEE Press, 2009, pp. 1243--1249. [Online].
  Available: \url{http://dl.acm.org/citation.cfm?id=1793974.1794188}
\BIBentrySTDinterwordspacing

\bibitem{6259860}
P.~Gopalan, C.~Huang, H.~Simitci, and S.~Yekhanin, ``On the locality of
  codeword symbols,'' \emph{IEEE Transactions on Information Theory}, vol.~58,
  no.~11, pp. 6925--6934, Nov 2012.

\bibitem{6284027}
D.~S. Papailiopoulos and A.~G. Dimakis, ``Locally repairable codes,'' in
  \emph{Information Theory Proceedings (ISIT), 2012 IEEE International
  Symposium on}, July 2012, pp. 2771--2775.

\bibitem{5934901}
F.~Oggier and A.~Datta, ``Self-repairing homomorphic codes for distributed
  storage systems,'' in \emph{INFOCOM, 2011 Proceedings IEEE}, April 2011, pp.
  1215--1223.

\bibitem{DBLP:journals/corr/OlmezR14}
\BIBentryALTinterwordspacing
O.~Olmez and A.~Ramamoorthy, ``Fractional repetition codes with flexible repair
  from combinatorial designs,'' \emph{CoRR}, vol. abs/1408.5780, 2014.
  [Online]. Available: \url{http://arxiv.org/abs/1408.5780}
\BIBentrySTDinterwordspacing

\bibitem{7458387}
M.~Y. Nam, J.~H. Kim, and H.~Y. Song, ``Locally repairable fractional
  repetition codes,'' in \emph{2015 Seventh International Workshop on Signal
  Design and its Applications in Communications (IWSDA)}, Sept 2015, pp.
  128--132.

\end{thebibliography}
\end{document}